\documentclass[a4paper,UKenglish]{lipics-v2018}

\usepackage{microtype}
\usepackage{amssymb} 
\usepackage{graphicx} 
\usepackage{amsthm} 
\usepackage{amsmath}
\usepackage{bbm}
\title{Percolation of Lipschitz surface and tight bounds on the spread of information among mobile agents}

\titlerunning{Percolation of Lipschitz surface and the spread of information among mobile agents}

\author{Peter Gracar}{Mathematical Institute, University of Cologne, Weyertal 86-90, 50931 Köln, Germany}{pgracar@math.uni-koeln.de}{}{}

\author{Alexandre Stauffer}{Department of Mathematical Sciences, University of Bath, Claverton Down, Bath, BA2 7AY, United Kingdom}{a.stauffer@bath.ac.uk}{}{Supported by a Marie Curie Career Integration Grant PCIG13-GA-2013-618588 DSRELIS, and an EPSRC Early Career Fellowship.}

\authorrunning{P. Gracar and A. Stauffer}

\Copyright{Peter Gracar and Alexandre Stauffer}

\subjclass{Mathematics of computing $\rightarrow$ Probability and statistics}


\category{}

\relatedversion{}

\supplement{}

\funding{}

\acknowledgements{}

\EventEditors{Eric Blais, Klaus Jansen, Jos\'{e} D. P. Rolim, and David Steurer}
\EventNoEds{4}
\EventLongTitle{Approximation, Randomization, and Combinatorial Optimization. Algorithms and Techniques (APPROX/RANDOM 2018)}
\EventShortTitle{APPROX/RANDOM 2018}
\EventAcronym{APPROX/RANDOM}
\EventYear{2018}
\EventDate{August 20--22, 2018}
\EventLocation{Princeton, NJ, USA}
\EventLogo{}
\SeriesVolume{116}
\ArticleNo{X} 
\nolinenumbers 
\hideLIPIcs  

\theoremstyle{plain} 
\newtheorem{thrm}{Theorem}
\newtheorem{claim}{Claim}

\theoremstyle{definition} 
\newtheorem{mydef}{Definition}

\begin{document}

\maketitle

\begin{abstract}
We consider the problem of spread of information among mobile agents on the torus. The agents are initially distributed as a Poisson point process on the torus, and move as independent simple 
	random walks. Two agents can share information whenever they are at the same vertex of the torus. 
	We study the so-called flooding time: the amount of time it takes for information to be known by all agents. We establish a tight upper bound on the flooding time, and introduce a technique 
	which we believe can be applicable to analyze other processes involving mobile agents.
 \end{abstract}

\section{Introduction}

We consider the problem of spread of information between mobile agents on a $d$-dimensional torus of side-length $n$.
We will denote by \(N=n^d\) the number of vertices on the torus, and will refer to the agents as \emph{particles}. 
At time 0, the particles are distributed on the vertices of the torus as a Poisson point process of intensity \(\lambda\). 
Then, particles move by performing independent continuous-time simple random walks on the torus; 
that is, at rate $1$ a particle chooses a neighboring vertex uniformly at random and jumps there.
It is not difficult to check that this system of particles is in stationarity. Thus, at any given time $t$, 
the location of the particles is a Poisson point process of intensity $\lambda$ 
on the torus. 
However, the configuration of particles at time $t$ is not independent of the configuration of particles at time $0$, 
and as we will explain below, it is this dependence that makes this model challenging to analyze. 

Assume that at time \(0\) there is a particle at the origin with a piece of information that has to be distributed to all other particles.
Then, any uninformed particle (a particle that does not know the information) receives the information whenever it is at the same vertex as an informed particle 
(a particle that knows the information). 
We study the time it takes the information to reach all the particles, which is commonly referred to as the \emph{flooding time}.
 
A big challenge in analyzing this model is due to the heavily dependent structure of the particles. 
In fact, though particles move independently of one another, dependences do arise over time. 
For example, if a ball of radius \(R\) centered at some vertex \(x\) of the torus turns out to have no particles at time 0, then the ball \(B(x,R/2)\) of radius \(R/2\) centered at \(x\) 
will continue to be empty of particles up to time \(R^2\), with positive probability. 
This means that the probability that the \((d+1)\)-dimensional, space-time cylinder \(B(x,R/2)\times[0,R^2]\) has no particle is at least \(\exp\{-cR^d\}\) for some constant \(c\).
This is just a stretched exponential on the volume of the cylinder, which prevents us from applying classical methods based on comparison with independent percolation~\cite{LSS}, since those require exponential 
decay of correlations. 
In addition to this, whenever one finds such a ball of radius $R$ empty of particles at time $0$, this affects regions of the torus in the vicinity of this ball. In particular, 
during a time interval of length $R^2$, the density of particle in the vicinity of the ball will be smaller than the expected density $\lambda$. 
In this work we develop a framework to control such dependences. 

When the transmission radius is large (in the sense that information can be transmitted between particles at distance \(O( \log^{1/d}(n))\) of each other)
or the jump range is large (in the sense that a particle can 
jump a distance of order $O(\log^{1/d} n)$ in one step), then the dependences can be more easily controlled. 
These cases where analyzed in~\cite{Clementi2011,Clementi2013},
where tight bounds on the flooding time (up to constant factors) were established\footnote{In fact, \cite{Clementi2011} studies the flooding time for a larger class of dynamic graphs. However, due to space limitations, we 
restrict our discussion to results on the specific model of spread of information among random walk particles.}. 
Having a large transmission radius or jump range helps the analysis because of the following. 
Tessellate the torus into boxes of side-length \(\Theta(\log^{1/d}(n))\), and tessellate time into intervals of constant length. 
Then, since the system of particles is in stationary and boxes are so large, 
we can apply a Chernoff bound for Poisson random variables 
to show that, for any given 
box and time interval, with probability $1-n^{-C}$, there is a large enough number of 
particles inside the box during that time interval (when this happen, call the cell of the tessellation \emph{good}). 
Then a union bound can be used to show that all cells of the tessellation are good.
Then, if the transmission radius is large enough to allow particles from neighboring boxes to exchange information, 
one can establish a tight bound on the flooding time. 
If it is the jump range that is large enough, then one can use the fact that, after a time interval of order 1, 
the configuration of particles inside any given box is close to stationarity. In other words, the system of particles has a small mixing time. 
This washes away the dependences of the system, and allowed a tight bound (up to constant factors) to be derived. 

An important open problem has been to analyze the case where both the transmission radius and the jump range are of order 1, which is our setting here. 
This was studied in \cite{Lam2011}, 
where it was shown that, with high probability, 
the flooding time is at most \(\tilde\Theta(n)\), 
where the notation $\tilde\Theta(\cdot)$ means that poly-logarithmic factors are neglected\footnote{We remark that \cite{Lam2011} considers also the case where the number of particles can be of order much smaller than $N$, 
and~\cite{PPPU,Peres2012} analyze a variant of this model, but these settings are out of the scope of this work.}. 
This bound is tight up to poly-logarithmic factors since, for a transmission radius and jump range of order 1, the flooding time is $\Omega(n)$ in all dimensions.
We note that, when neglecting poly-logarithmic factors, one can still work with the above tessellation --- of cells of side-length \(\Theta(\log^{1/d}(n))\) --- 
for which all cells of the tessellation are good. This is because one can do some suboptimal estimates to allow information to spread inside a cell (thereby 
losing only a poly-logarithmic factor), and then use the fact 
that cells are good, and full of particles, to let the information spread from one cell to the next.
Getting a bound that is tight up to \emph{constant} factors, on the other hand, involves a rather delicate issue, 
since one is forced to consider tessellations of \emph{constant} side-length, which will naturally
contain a positive density of bad cells, forcing a more careful control of the dependences of the system.

Turning back to our setting, where particles can jump only across neighboring vertices and information can be transmitted only between particles located at the same vertex, 
\cite{Kesten2005} analyzes the process in the whole of $\mathbb{Z}^d$ and shows that the information spreads with positive speed. 
To prove this, the authors developed a complicated multi-scale framework to control the dependences of the system, 
where tessellations of different side-lengths were considered and controlled. 
This multi-scale technique is quite powerful, and has been employed in the mathematics literature to solve other processes with slow decay of 
correlations~\cite{Sidoravicius2009,Sznitman2012,Candellero2015}. 
However, this technique is usually very difficult to implement, 
and has to be tailored to each specific model and question being studied.
The goal of our work is to develop a robust and flexible multi-scale framework that can be more easily applied to answer questions involving systems of random walk particles, and we illustrate its usefulness by deriving 
tight bounds on the flooding time.

\subsection{Our results}
We start considering a more general setup.
Let \(\mathbb{T}^d\) be the \(d\)-dimensional integer torus of side length \(n\).
Let \(G=(\mathbb{T}^d,E)\) be the nearest neighbor graph on \(\mathbb{T}^d\). Let \(\{\mu_{x,y}\}_{(x,y)\in E}\) be a collection of i.i.d.\ symmetric weights, which we call \emph{conductances}. 
We assume that the conductances are \emph{uniformly elliptic}; that is, 
\begin{equation}\label{eq:mu_bounds_new}
	\textrm{there exists a constant \(C_M>0\), such that }\mu_{x,y}\in[C_M^{-1},C_M]\textrm{ for all }(x,y)\in E.
\end{equation}
We say \(x\sim y\) if \((x,y)\in E\) and define \(\mu_x=\sum_{y\sim x}\mu_{x,y}\). At time \(0\), consider a Poisson point process of particles on \(\mathbb{T}^d\), 
with intensity measure \(\lambda(x)=\lambda_0\mu_x\) for some constant \(\lambda_0>0\) and all \(x\in\mathbb{T}^d\). 
That is, for each \(x\in\mathbb{T}^d\), the number of particles at \(x\) at time \(0\) is an independent Poisson random variable of mean \(\lambda_0\mu_x\). 
Then, let the particles perform independent continuous-time simple random walks on the weighted graph so that a particle at \(x\in\mathbb{T}^d\) jumps to a neighbor \(y\sim x\) at rate \(\frac{\mu_{x,y}}{\mu_x}\). 
It follows from the thinning property of Poisson random variables that the system of particles is in stationarity.

Assume that at time \(0\) there is an informed particle at the origin, and all other particles are uninformed. 
One of the main results of this paper is the following. 
\begin{thrm}\label{thrm:total}
   If \(d\geq 2\) and the conductances satisfy (\ref{eq:mu_bounds_new}), then with probability $1-n^{-\omega(1)}$ the flooding time is $\Theta(n)$.
\end{thrm}

Another main contribution of this paper is the framework we develop to establish Theorem~\ref{thrm:total}, which we believe gives a robust and more 
easy to apply framework to address problems involving systems of moving particles. The idea is as follows.
We tessellate space and time into cells of constant length. 
Then, for each cell we are given a local event, and call the cell good if the event of that cell 
holds. 
Then, if for any given cell, we have that the probability that the cell is good is close enough to 1, then we can find a subset of good cells that form what we call 
a Lipschitz surface and a Lipschitz net. 
These Lipschitz surface and Lipschitz net have some percolative and geometric features that allow 
the good event to propagate through space and time. 
For example, for the problem of spread of information, the local event we use is to say that a given cell is good 
if the following two things happen: (i) the cell contains sufficiently many particles, and 
(ii) if there is an informed particle inside the cell, then that particle is able to inform a large number of other particles 
that will move to neighboring cells. With this definition and the existence of the Lipschitz surface and net, 
we obtain that once the information enters a cell of the Lipschitz surface, we guarantee that the information can 
propagate throughout the surface, from one cell of the surface to the next.
We believe our approach is flexible enough to allow other processes on moving particles to be analyzed. The main task reduces to defining a suitable local event. 

Since this framework is quite involved, we will give its construction and all main technical theorems in Section~\ref{sec:lipschitz}. 
Then, in Section~\ref{sec:spread}, we use this framework to analyze the spread of information.
Due to space limitations, we will not be able to give full proofs of the above framework, for which we refer to the full version~\cite{Gracar2016a}.
This extended abstract has yet one additional result with respect to~\cite{Gracar2016a}, which is the construction and proof of the Lipschitz net, which is adapted to analyzing processes on finite graphs.

\section{Lipschitz net framework}\label{sec:lipschitz}

For the remainder of this paper, we assume \(d\geq 2\). Fix \(\ell>0\) and tessellate \(\mathbb{T}^d\) into cubes of side length \(\ell\in\mathbb{R}\), indexed by \(i\in \mathbb{Z}^{d}\).
To simplify the notation, assume that $n/\ell$ is an integer.
Next, tessellate time into intervals of length \(\beta\), indexed by \(\tau\in\mathbb{Z}\). 
With this we denote by the \emph{space-time cell} \((i,\tau)\in\mathbb{Z}^{d+1}\) the region \(\prod_{j=1}^d[i_j\ell,(i_j+1)\ell]\times[\tau\beta,(\tau+1)\beta]\). 
In the following, $\beta$ and $\ell$ are constants such that the ratio \(\beta/\ell^2\) is fixed first to be some small value, and then later \(\ell\) is made large enough. 
We will also need to consider overlapping space-time cells. 
Let \(\eta\geq 1\) be an integer which will represent the amount of overlap between cells. For each cube \(i=(i_1,\dots,i_d)\) and time interval \(\tau\), define the \emph{super cube} \(i\) as \(\prod_{j=1}^d[(i_j-\eta)\ell,(i_j+\eta+1)\ell]\) and the \emph{super interval} \(\tau\) as \([\tau\beta,(\tau+\eta)\beta]\). We define the \emph{super cell} \((i,\tau)\) as the Cartesian product of the super cube \(i\) and the super interval \(\tau\).

For any time $s$, let $\Pi_s$ be the set of particles at time $s$, seen as a collection of vertices of $G$ with multiplicity when there is more than one particle at a vertex. 
We say an event \(E\) is \emph{increasing} for \((\Pi_s)_{s\geq 0}\) if the fact that \(E\) holds for \((\Pi_s)_{s\geq 0}\) implies that it holds for all \((\Pi'_s)_{s\geq 0}\) for which \(\Pi_s'\supseteq \Pi_s\) for all \(s\geq 0\).
We need the following definitions.

\begin{mydef}\label{def:restricted}
	We say an event \(E\) is \emph{restricted} to a region \(X\subset\mathbb{T}^d\) and a time interval \([t_0,t_1]\) if it is measurable with respect to the \(\sigma\)-field generated by all the particles that are inside \(X\) at time \(t_0\) and their positions from time \(t_0\) to \(t_1\).
\end{mydef}
\begin{mydef}\label{def:displacement}
	We say a particle has displacement inside \(X'\) during a time interval \([t_0,t_0+t_1]\), if the location of the particle at all times during \([t_0,t_0+t_1]\) is inside \(x+X'\), where \(x\) is the location of the particle at time \(t_0\).
\end{mydef}

\begin{mydef}\label{def:probassoc}
	\(\nu_E\) is called the \emph{probability associated} to an increasing event \(E\) that is restricted to \(X\)  and a time interval \([0, t]\) if, 
	for an intensity measure \(\zeta\) and a region \(X'\in\mathbb{T}^d\), \(\nu_E(\zeta,X,X',t)\) is the probability that \(E\) happens given that, at time \(0\), 
	the particles in \(X\) are distributed as a Poisson point process of intensity \(\zeta\)
	and their motions from \(0\) to \(t\) are independent continuous time random walks on the weighted graph \((G,\mu)\), where the particles are conditioned to have displacement inside \(X'\) during \([0,t]\).
\end{mydef}

For each \((i,\tau)\in\mathbb{T}^{d}\times \mathbb{Z}\), let \(E_{\mathrm{st}}(i,\tau)\) be an increasing event restricted to the super cube \(i\) and the super interval \(\tau\). Here the subscript \(\mathrm{st}\) refers to space-time. 
We say that a cell \((i,\tau)\) is \emph{bad} if \(E_{\mathrm{st}}(i,\tau)\) does not hold; otherwise, \((i,\tau)\) is called \emph{good}.

Our framework will establish that if for any given $(i,\tau)$, the event $E_{\mathrm{st}}(i,\tau)$ occurs with large enough probability, then not only do the good cells percolate but the good cells form a particularly 
useful geometry, which we will call the Lipschitz net. 

Before defining the Lipschitz net, we need to introduce a different way to index space-time cells, which we refer to as the \emph{base-height index}. 
In the base-height index, we pick one of the \(d+1\) space-time dimensions and denote it as \emph{height}, using index \(h\in\mathbb{Z}\), while the remaining \(d\) space-time dimensions will form the base, 
which will be indexed by \(b\in\mathbb{Z}^d\). 
In this way, for each space-time cell \((i,\tau)\) there will be \((b,h)\in\mathbb{Z}^{d+1}\) such that the base-height cell \((b,h)\) corresponds to the space-time cell \((i,\tau)\). With this, 
we set \(E_\mathrm{bh}(b,h)=E_\mathrm{st}(i,\tau)\). (Here the subscript \(\mathrm{bh}\) refers to base-height.)
It might be tempting to choose time as the height dimension, however it turns out that selecting one of the spatial dimensions to act as height is a better choice, as will be shown below.
With this choice, note that $b\in \mathbb{T}^{d-1}\times \mathbb{Z}$ and $h\in \mathbb{T}$; thus, for notation purpose, we define $\mathbb{T}^{d}_* = \mathbb{T}^{d-1}\times \mathbb{Z}$
and $\mathbb{T}^{d+1}_* = \mathbb{T}^{d-1}\times \mathbb{Z}\times \mathbb{T}$.

\subsection{Two-sided Lipschitz surface}
\begin{mydef}\label{def:lip_fun}
	A function \(F:\mathbb{T}^d_*\rightarrow \mathbb{T}\) is called a \emph{Lipschitz function}
	 if \(|F(x)-F(y)|\leq 1\) whenever \(\|x-y\|_1 = 1\).
\end{mydef}

\begin{mydef}\label{def:lip_surf}
	A \emph{two-sided Lipschitz surface} \(F\) is a set of base-height cells \((b,h)\in\mathbb{T}^{d+1}_*\) such that for all \(b\in\mathbb{T}^d_*\) there are exactly two (possibly equal) integer values \(F_+(b)\geq 0\) and \(F_-(b)\leq0\) for which \((b,F_+(b)),(b,F_-(b))\in F\) and, moreover, \(F_+\) and \(F_-\) are Lipschitz functions.
\end{mydef}

\begin{figure}[!h]
  \begin{center}
  	\includegraphics[width=0.8\linewidth]{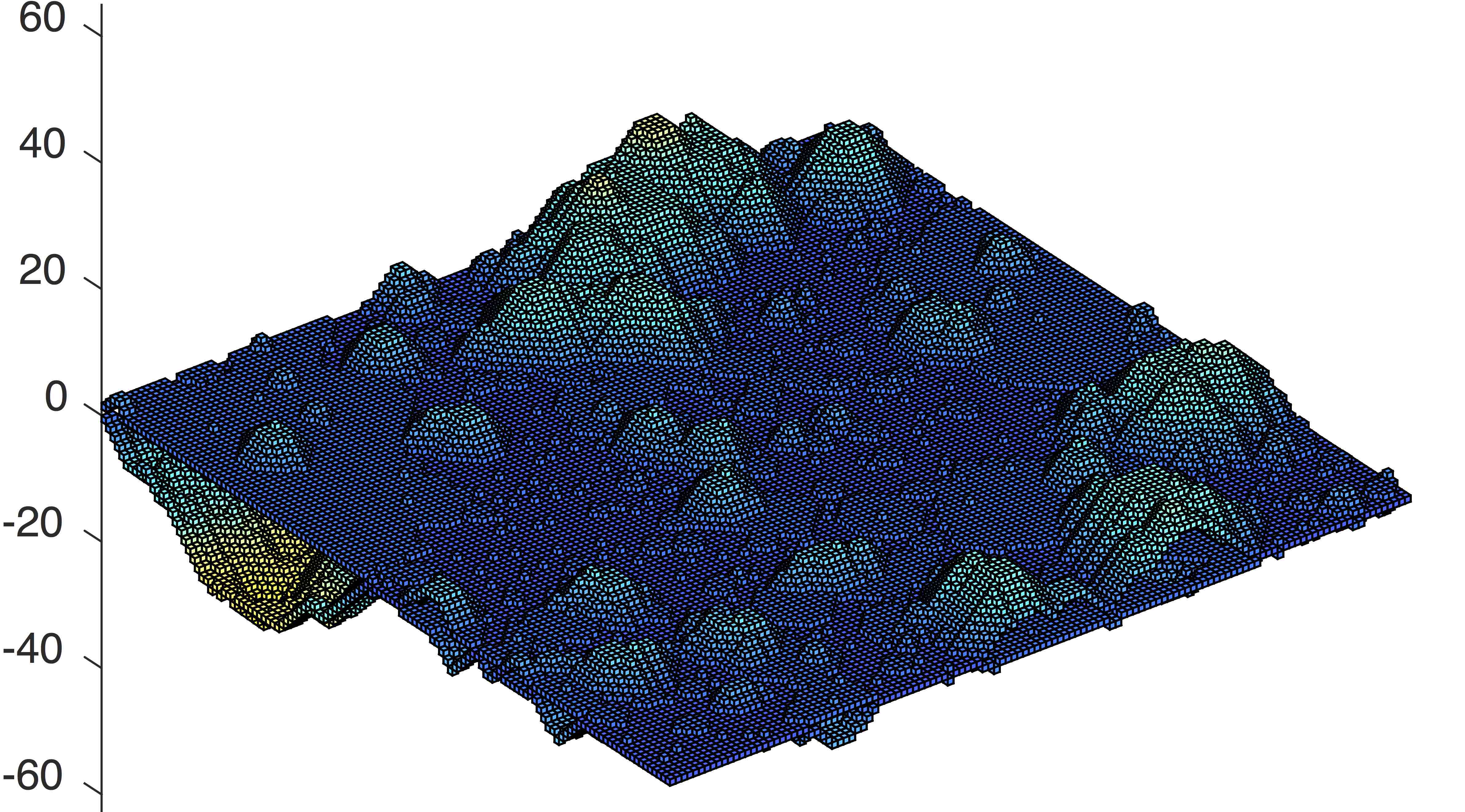}
  \end{center}
  \caption{A realization of the two-sided Lipschitz surface for the case \(d=2\).}\label{fig:surface}
\end{figure}

We say a space-time cell \((i,\tau)\) belongs to \(F\) if its corresponding base-height cell \((b,h)\) belongs to \(F\). 
For a positive integer $D$, we say a two-sided Lipschitz surface \emph{surrounds} a cell \((b',h')\) at distance \(D\) 
if any path \((b',h')=(b_0,h_0),(b_1,h_1),\dots,(b_m,h_m)\) for which \(\|(b_i,h_i)-(b_{i-1},h_{i-1})\|_1=1\) for all \(i\in\{1,\dots m\}\) and \(\|(b_m,h_m)-(b_0,h_0)\|_1>D\), intersects with \(F\).

For any \(z\in\mathbb{Z}_+\), define the cube \(Q_z=[-z/2,z/2]^d\). 
The following theorem establishes the existence of the Lipschitz surface. Due to space limitations, the proof is given in~\cite{Gracar2016a}.

\begin{thrm}\label{thrm:surface_event_simple}
   Consider the graph \((G,\mu)\) satisfying \eqref{eq:mu_bounds_new}, and the tessellation defined above. 
   There exist positive constants \(c_1, c_2, c_3, c_4\) and \(c_5\) such that, if \(\beta/\ell^2\leq c_5\), then the following holds.
	Let \(E_{\mathrm{st}}(i,\tau)\) be any increasing event restricted to the space-time super cell \((i,\tau)\).
	Fix \(\epsilon\in(0,1)\) and fix \(w\geq\sqrt{\frac{\eta\beta}{c_2\ell^2}\log\left(\frac{8c_1}{\epsilon}\right)}.\)
	Then, there exists a positive number \(\alpha_0\) that depends on \(\epsilon\), \(\eta\) and the ratio \(\beta/\ell^2\) so that if
	\begin{equation}
	\min\left\{C_{M}^{-1}\epsilon^2\lambda_0\ell^d,\log\left(\frac{1}{1-\nu_{E_{\mathrm{st}}}((1-\epsilon)\lambda,Q_{(2\eta+1)\ell},Q_{w\ell},\beta)}\right)\right\}\geq\alpha_0,
	\label{eq:lip}
	\end{equation}
	a two-sided Lipschitz surface \(F\) where \(E_{\mathrm{st}}(i,\tau)\) holds for all \((i,\tau)\in F\) exists almost surely, and 
	the probability that $F$ does not surround the origin at distance $r$ is at most 
	\begin{align*}
		\begin{array}{ll}
			\sum_{s\geq r}s^d\exp\left\{-c_3\lambda_0\frac{\ell s}{\log^{c_4}(\ell s)}\right\},&\textrm{for }d=2\\
			\sum_{s\geq r}s^d\exp\left\{-c_3\lambda_0\ell s\right\},&\textrm{for }d\geq 3.
		\end{array}
	\end{align*}
\end{thrm}

\begin{remark}
	The proofs in \cite{Gracar2016a} give the existence of the two-sided Lipschitz surface on the whole of \(\mathbb{Z}^d\), but the very same proof works for the torus.
\end{remark}

\begin{remark}
Theorem~\ref{thrm:surface_event_simple} is key to our framework.
We now briefly explain how it can be used. 
The event $E_\mathrm{st}$ can be any local event, where in $\nu_{E_\mathrm{st}}$, $Q_{(2\eta+1)\ell}$ gives the region on which the event is measurable. 
To control dependences, we consider the larger cube $Q_{w\ell}$, inside which the particles that start from $Q_{(2\eta+1)\ell}$ are conditioned to stay during the time interval $\beta$. 
Then $w$ has to be large enough, as specified in the theorem, so that this conditioning is likely to happen.
Then, $\nu_{E_\mathrm{st}}$ gives the probability that the event happens given that the initial configuration of particle is a Poisson point process of intensity measure 
$(1-\epsilon)\lambda$, just slightly smaller than the intensity measure $\lambda$ we started with. We disregard an ``$\epsilon$-fraction of the particles'' because naturally, in any given space-time cell, 
some particles move atypically and 
will not 
be organized exactly as a Poisson point process; but those particles can be neglected using the assumption that $E_\mathrm{st}$ is increasing. 
Thus~\eqref{eq:lip} requires that $\nu_{E_\mathrm{st}}$ is at least $1-\exp(-\alpha_0)$ for a Poisson point process of intensity $(1-\epsilon)\lambda$. 
This is usually achievable by properly defining the event $E_\mathrm{st}$ to be such that its occurrence increases 
with $\ell$ (the size of the tessellation).
Then, \eqref{eq:lip} also requires that $C_{M}^{-1}\epsilon^2\lambda_0\ell^d\geq \alpha_0$. After fixing $\epsilon$,
this can be satisfied either by setting $\ell$ large enough or by assuming that the constant $\lambda_0$ governing the density of particles is large enough. 
This condition is natural in applications: one either requires the size of cells to be large (which will be the case in our application for the flooding time) or the tessellation is more restricted (for example,
limited to the transmission radius of the particles) and one requires the density of particles to be large enough, as in~\cite{Stauffer2014}.
\end{remark}

\subsection{Lipschitz net}
We are now ready to define the \emph{Lipschitz net} on the torus \(\mathbb{T}^d\) that we will use to prove Theorem \ref{thrm:total}. 
The Lipschitz net, roughly speaking, will be an interlacement of Lipschitz surfaces, where we will take each spatial coordinate as being height in the base-height index, 
and for each of them we will have a pile of surfaces. More formally, 
let \(k\in\left\{0,1,\dots,\left\lfloor\frac{n}{\ell\log^3(n/\ell)}\right\rfloor\right\}\) and \(q\in\{1,2,\dots,d\}\). 
For any $i\in\mathbb{T}^d$, let $i=(i_1,i_2,\ldots,i_d)$.
Define \(\mathbb{L}_k^q\) to be the \(d\)-dimensional 
hyperplane on the space-time tessellation that is orthogonal to the \(q\)-th spatial coordinate, with distance from the origin of \(k\left\lceil\log^3(n/\ell)\right\rceil\), i.e.
\[
	\mathbb{L}_k^q=\left\{(i,\tau)\in\mathbb{T}^{d}\times\mathbb{Z}:\:i_{q}=k\lceil\log^3(n/\ell)\rceil\right\}.
\]
We define \(F_k^q\) to be the Lipschitz surface corresponding to \(\mathbb{L}_k^q\), i.e.\ 
\(F_k^q\) is the two-sided Lipschitz surface for which \(h\) in the base-height index corresponds to \(i_q\) in the space-time index, and for which the Lipschitz functions satisfy \(F_+(b)\geq k\lceil\log^3(n/\ell)\rceil\) and \(F_-(b)\leq k\lceil\log^3(n/\ell)\rceil\). We define the \emph{height of the surface} \(F_k^q\) at \(b\in \mathbb{Z}^d\) to be 
	\[
		\max_{h:(b,h)\in F_k^q}\left|k\lceil\log^3(n/\ell)\rceil-h\right|.
	\] 

Let \(C_0>0\) be an integer constant of our choosing. From now on we assume that \(F_k^q\) is a Lipschitz surface for which the height is at most \(\frac{\log^3(n/\ell)}{2}\) for all \((i,\tau)\in F_k^q\) satisfying \(\tau\in\{0,1,\dots,C_0n/\ell\}\).

\begin{mydef}
	The \emph{Lipschitz net} \(F_{\textrm{net}}\) with constant \(C_0\) is the set of space-time cells \((i,\tau)\in\mathbb{T}^{d}\times\mathbb{Z}\) contained in the union of all \(F_k^q\); i.e, 
	$
		F_\mathrm{net}=\bigcup_{q=1}^d\bigcup_{k=0}^{\left\lfloor\frac{n}{\ell\log^3(n/\ell)}\right\rfloor} F_k^q.
	$
	Moreover, we say that $F_\mathrm{net}$ surrounds the origin at distance $D$ if $F_0^q$ surrounds the origin of $\mathbb{L}_0^q$ at distance $D$ for all $q\in\{1,2,\ldots,d\}$.
\end{mydef}

Note that we have for all \((i,\tau)\in F_{\textrm{net}}\) that the event \(E_{\textrm{st}}(i,\tau)\) holds, which follows directly from the fact that every space-time cell in \(F_{\textrm{net}}\) belongs to at least one Lipschitz surface \(F_k^q\) for some \(k\) and some \(q\).

\begin{thrm}\label{thrm:net} 
   For any constant $C_0$, there exist a constant $C_1>0$ such that, for any $\delta>0$ and any $\ell=O(n^{1-\delta})$ with $\ell\geq C_1$, 
	the Lipschitz net $F_{\textrm{net}}$ with constant $C_0$ exists and surrounds the origin at distance $O(\log^2 n)$ with probability $1-n^{-\omega(1)}$.
\end{thrm}
\begin{proof}
Start by considering the plane \(\mathbb{L}_0^1\) and its corresponding Lipschitz surface \(F_0^1\). 
If the height of \(F_0^1\) at the origin is more than \(\frac{\log^3(n/\ell)}{2}\), 
then the Lipschitz surface cannot surround the origin at a distance \(\frac{\log^3(n/\ell)}{2}\). Therefore, since \(\ell\) and \(\log^3(n/\ell)\) are both assumed sufficiently large, we have
by Theorem \ref{thrm:surface_event_simple} that the probability that a two-sided Lipschitz surface around the origin with height at most \(\frac{\log^3(n/\ell)}{2}\) does not exists is at most 
\[
	\sum\nolimits_{s\geq\log^3(n/\ell)/2}s^d\exp\left\{-C\lambda_0\frac{\ell s}{\log^c(\ell s)}\right\}
	\leq \exp\left(-\omega(\log^2n)\right).
\]
Using this and a uniform bound across all space-time cells for which \(\tau\in\{0,1,\dots,C_0n/\ell\}\), we have that the probability that \(F_0^1\) has height at most \(\frac{\log^3(n/\ell)}{2}\) 
for all \((i,\tau)\in F_k^q\) satisfying \(\tau\in\{0,1,\dots,C_0n/\ell\}\), is at least $1-\exp\left(-\omega(\log^2n)\right)$.

Next, consider the planes \(\mathbb{L}_k^q\). Since the probability space is translation invariant due to the weights \(\mu_{x,y}\) being i.i.d., this bound holds for any \(k\) and any $q$.
Therefore, by applying a uniform bound across \(k\in\left\{0,1,\dots,\left\lfloor\frac{n}{\ell\log^3(n/\ell)}\right\rfloor\right\}\) and \(q\in\{1,2,\dots,d\}\) 
we obtain that the probability that \(F_k^q\) has maximum height at most \(\frac{\log^3(n/\ell)}{2}\) for all $k$ and $q$ is at least $1-\exp\left(-\omega(\log^2n)\right)$.
Under this assumption, for any given $q$ and two distinct $k,k'$, the surfaces $F_k^q$ and $F_{k'}^q$ do not intersect, producing the Lipschitz net.
\end{proof}

The usefulness of the Lipschitz net is that, once we know it exists for any local event $E_\mathrm{st}$ that is likely enough, 
then one just needs to find a suitable choice for the event $E_\mathrm{st}$ and use the Lipschitz net to show that this event propagates throughout the torus. 
For the case of spread of information, we will use the Lipschitz net to show that once an informed particle enters a cell that is part of the Lipschitz net, then 
information spreads evenly across the torus resulting in a density of informed particles. 
For this, we will use a specific increasing event \(E_{\textrm{st}}\) to obtain that the information spreads with positive speed on each individual surface of \(F_{\textrm{net}}\). 
Then, in order to show that the information also moves across different surfaces of the net, we will need the following geometric property.

\begin{lemma}\label{lem:change_surface}
	Let \(F_{\textrm{net}}\) be the Lipschitz net with constant \(C_0\) and let
	\( F_{k}^{q}\) and \( F_	{k'}^{q'}\) be any two given Lipschitz surfaces that are part of \(F_{\textrm{net}}\), 
	where \(q\neq q'\) and \(k,k'\in\left\{0,1,\dots,\left\lfloor\frac{n}{\ell\log^3(n/\ell)}\right\rfloor\right\}\). 
	For any \(\tau\in\{0,1,\dots,C_0n/\ell\}\) there exist space-time cells \((i,\tau)\in  F_{k}^{q}\) and \((i',\tau)\in  F_{k'}^{q'}\) such that \(\|(i',\tau)-(i,\tau)\|_{1}\leq 1\).
\end{lemma}

\begin{proof}
	Let \(q=1\) and \(q'=2\); the proof for other combinations of parameters \(q\) and \(q'\) goes similarly.
	We want to show that for any \(\tau\in\{0,1,\dots,C_0n/\ell\}\) there exist a space-time cell \((i_1,\dots,i_{d},\tau)\in  F_k^1\) and a 
	space-time cell \((j_1,\dots,j_{d},\tau)\in  F_{k'}^2\) such that \(\|(i_1,\dots,i_{d})-(j_1,\dots,j_{d})\|_{1}\leq 1\).
	Fix \(\tau\in\{0,1,\dots,C_0n/\ell\}\) and set the components \((i_3,\dots,i_{d})\) to be the same as \((j_3,\dots,j_d)\). 

	Let \(F^1\) be either of the two Lipschitz functions (see Definition \ref{def:lip_surf}) corresponding to \( F_k^1\),  
	and let \(F^2\) be either of the Lipschitz functions corresponding to \( F_{k'}^2\). 
	Since \((i_3,\dots,i_{d})=(j_3,\dots,j_d)\), to simplify notation we 
	write \(F^1(y):=F^1(y,i_3,\dots,i_d,\tau)\in\mathbb{T}\) and 
	\(F^2(y):=F^2(y,i_3,\dots,i_d,\tau)\in\mathbb{T}\). 
	Therefore it remains to show that there exists $x,y\in\mathbb{T}$ such that 
	$|(F^1(x),x)-(y,F^2(y))|\leq 1$. Assume, by contradiction, that this is not the case.
	
	Let $m^2=k'\lceil\log^3(n/\ell)\rceil$, which is the height of $\mathbb{L}_{k'}^2$, that is, $(0,m^2,0,0,\ldots,0)\in \mathbb{L}_{k'}^2$.
	Next, if $F^2$ is the Lipschitz function corresponding to $F_+$ of $F_{k'}^2$ (refer to Definition~\ref{def:lip_surf}) then set $h^2=m^2+\frac{\log^3(n/\ell)}{2}+1$; otherwise,
	set $h^2=m^2-\frac{\log^3(n/\ell)}{2}-1$. So for all $y\in\mathbb{T}$, we have $|F^2(y)-m^2| \leq |F^2(y)-h^2|$.
	
	For any point $(x,y)\in\mathbb{T}^2$ we say that it is \emph{under} $F^2$ if $|y-m^2|\leq |F^2(x)-m^2|$; otherwise we say it is above $F^2$.
	Note that $(F^1(m^2),m^2)$ is ``under'' the surface $F^2$, and $(F^1(h^2),h^2)$ is above $F^2$. Therefore, we take the shortest sequence $x_1,x_2,\ldots,x_\iota$ from $m^2$ to $h^2$ there must 
	exist a point $x_r$ such that $(F^1(x_r),x_r)$ is under $F^2$ but $(F^1(x_{r+1}),x_{r+1})$ is above $F^2$. Since $F^1$ is Lipschitz, this implies that one of these two points is within distance $1$ from $F^2$.
\end{proof}

\section{Spread of information using the Lipschitz net}\label{sec:spread}
We proceed to showing how the information spreads on \(F_{\textrm{net}}\). We do this by applying
Theorem \ref{thrm:net} with an event that results in the information spreading with positive speed along each individual Lipschitz surface of \(F_{\textrm{net}}\). More precisely, from now on let the increasing event \(E_{\textrm{st}}(i,\tau)\) be defined as below in Definition \ref{def:mainEvent}.

\begin{mydef}[Increasing event \(E_{\textrm{st}}(i,\tau)\)]\label{def:mainEvent}
   Take any $(i,\tau)\in\mathbb{T}^d\times \mathbb{Z}$. Let \(\Upsilon\) be the collection of particles located inside $\prod_{j=1}^d[(i_j-\eta)\ell,(i_j+\eta+1)\ell]$ at time $\tau\beta$.
   Consider a distinguished particle $x_0$ located inside $\prod_{j=1}^d[i_j\ell,(i_j+1)\ell]$ at time $\tau\beta$.
	Define \(E_{\mathrm{st}}(i,\tau)\) to be the event that at time \((\tau+1)\beta\), for all \(i'\in\mathbb{T}^d\) with \(\|i-i'\|_{\infty}\leq\eta\), 
	there is at least one particle from \(\Upsilon\) in \(\prod_{j=1}^d[i_j'\ell,(i_j'+1)\ell]\) that collided with \(x_0\) during \([\tau\beta,(\tau+1)\beta]\).
\end{mydef}
For \(E_{\textrm{st}}(i,\tau)\) defined as above, we have the following result.
The proof of this result uses a few heat-kernel estimates for random walks on $\mathbb{Z}^d$ with i.i.d.\ conductances.

\begin{lemma}\label{prop:event}
	Fix any $\epsilon$, $\eta$ and the ratio $\beta/\ell^2$. Let $w$ satisfy the condition in
	Theorem~\ref{thrm:surface_event_simple}. Then, if \(\ell\) is sufficiently large,
	then there exists a positive constant \(C\) such that for \(E_{\textrm{st}}(i,\tau)\) as defined in Definition \ref{def:mainEvent} and for any $(i,\tau)\in\mathbb{T}^d\times \mathbb{Z}$, we have
	$
		\nu_{E_{\mathrm{st}}}((1-\epsilon)\lambda,Q_{(2\eta+1)\ell},Q_{w\ell},\beta)\geq1-\exp\{-C(1-\epsilon)\lambda_0\ell^{1/3}\}.
	$
\end{lemma}

\begin{proof}
Let \(T=\ell^{5/3}\).
Since $\beta/\ell^2$ is fixed, we can set $\ell$ a large enough constant so that $T \ll \beta$ (i.e., $T$ is much smaller than the length of the time interval in the tessellation).
Define \(Q^*:=\prod_{j=1}^d[(i_j-\eta)\ell,(i_j+\eta+1)\ell]\) and assume that at time \(\tau\beta\), for all sites \(x\in Q^*\), the number of particles at \(x\) is a Poisson random variable with mean \((1-\epsilon)\lambda_0\mu_x\). 

We start by stating two two claims and using them to prove the lemma. Then, we give the proof of the claims.  

\begin{claim}\label{cl:1}
   If the distinguished particle \(x_0\) is inside \(\prod_{j=1}^d[i_j\ell,(i_j+1)\ell]\) at time \(\tau\beta\) and \(x_0\) follows a fixed path \(\left(\rho(t)\right)_{\tau\beta\leq t\leq \tau\beta+T}\), then by time \(\tau\beta+T\) the number of particles that have collided with \(x_0\) during \([\tau\beta,\tau\beta+T]\), but were not at the same site as \(x_0\) at time \(\tau\beta\), is a Poisson random variable with intensity at least \(
		C_1(1-\epsilon)\lambda_0\ell^{1/3}
	\) for some positive constant \(C_1\), independent of \(\rho(t)\).
\end{claim}

\begin{claim}\label{cl:2}
   Given that there are \(N\) particles inside of \(Q^*\) at time \(\tau\beta+T\), the probability that at least one of these particles is inside \(Q^{**}:=\prod_{j=1}^d[(i'_j)\ell,(i'_j+1)\ell]\) for any \(i'\) for which \(|i-i'|\leq\eta\) is at least \(
	1-\exp\{-Nc_p\},
	\)
	where \(c_p\) is a positive constant that is bounded away from \(0\) and depends only on \(d\), \(\eta\) and the ratio \(\beta/\ell^2\).
\end{claim} 

Now we use the above claims to prove the lemma. 
%
	Note that by Definition \ref{def:mainEvent}, \(E_{\textrm{st}}(i,\tau)\) is restricted to the super cube \(Q^*\) and time interval \([\tau\beta,(\tau+1)\beta]\).
	We now define the following 3 events.
	\begin{description}
		\item[\(F_1\):] The distinguished particle \(x_0\) never leaves \(\prod_{j=1}^d[(i_j-\eta+1)\ell,(i_j+\eta-1)\ell]\) during \([\tau\beta,\tau\beta+~T]\).
		\item[\(F_2\):] Let \(C_1\) be the constant from Claim~\ref{cl:1}. During the time interval \([\tau\beta,\tau\beta+T]\) the distinguished particle \(x_0\) collides with at least \(\frac{C_1\lambda_0\ell^{1/3}}{2}\) different particles from \(\Upsilon\) that are in the super cube \(Q^{*}\) at time \(\tau\beta+T\).
		\item[\(F_3\):] Out of the \(\frac{C_1\lambda_0\ell^{1/3}}{2}\) or more particles from \(F_2\), at least one of them is in the cube \(Q^{**}\) at time \((\tau+1)\beta\), for all \(Q^{**}\) for which \(Q^{**}\subset Q^{*}\).
	\end{description}
	By definition of the events, we clearly have that \(\mathbb{P}[E_{\mathrm{st}}(i,\tau)]\geq \mathbb{P}[F_1\cap F_2\cap F_3]\). Also note that \(F_1,F_2\) and \(F_3\) are clearly restricted to the super cube \(Q^*\) and the time interval \([\tau\beta,(\tau+1)\beta]\) and are all increasing events.

	Using the exit probability bound from \cite[Proposition 3.7]{Barlow2004} we have
	\begin{equation}\label{for:F1}
		\mathbb{P}[F_1]\geq 1-C_2\exp\{-C_3\ell^2/T\}=1-C_2\exp\{-C_3\ell^{1/3}\}
	\end{equation}
	for some positive constants \(C_2\) and \(C_3\).

	For the event \(F_2\), we apply the result of Claim~\ref{cl:1}. Note that the bound from Claim~\ref{cl:1} is uniform across all paths \(\rho(\cdot)\) and in particular holds for any path the distinguished particle from the event \(F_1\) might follow. This gives that the intensity of the Poisson point process of particles that are in \(Q^{*}\) at time \(\tau\beta\) and collide with \(x_0\) during \([\tau\beta,\tau\beta+T]\) is at least \((1-\epsilon)\lambda_0 C_1\ell^{1/3}\) for some positive constant \(C_1\). Since every particle that collides with \(x_0\) enters \(\prod_{j=1}^d[(i_j-\eta+1)\ell,(i_j+\eta)\ell]\) during \([\tau\beta,\tau\beta+T]\), we can again use the exit probability bound from \cite[Proposition 3.7]{Barlow2004} to bound the probability that the particle is outside of \(Q^{*}\) at time \(\tau\beta+T\) from below by
	\[
		1-C_a\exp\left\{-\frac{C_b\ell^2}{T}\right\}=1-C_a\exp\{-C_b\ell^{1/3}\},
	\]
	for some positive constants \(C_a\) and \(C_b\). This term can be made as close to \(1\) as possible by having \(\ell\) sufficiently large. We assume \(\ell\) is large enough so that this term is larger than \(2/3\). This gives that the intensity of the process of particles from \(\Upsilon\) that collided with \(x_0\) during \([\tau\beta,\tau\beta+T]\) and are in \(Q^{*}\) at time \(\tau\beta+T\) is at least 
	\[
		\frac{2(1-\epsilon)\lambda_0 C_1\ell^{1/3}}{3}.
	\]
	Using Chernoff's bound (see Lemma \ref{lem:chernoff}) we have that
	\begin{equation}\label{for:F2}
		\mathbb{P}[F_2]\geq 1-\exp\{-(2/3)^2C_1(1-\epsilon)\lambda_0\ell^{1/3}\}.
	\end{equation}
	
	We now turn to \(F_3\). Using the result of Claim~\ref{cl:2}, and a uniform bound across the number of cubes inside a super cube, we have that 
	\begin{equation}\label{for:F3}
		\mathbb{P}[F_3]\geq 1-(2\eta+1)^d\exp\left\{-\frac{C_1(1-\epsilon)\lambda_0\ell^{1/3}}{2}c_p\right\},
	\end{equation}
	where \(c_p\) is a small but positive constant. Taking the product of the probability bounds in (\ref{for:F1}), (\ref{for:F2}) and (\ref{for:F3}), we see that the probability that \(E_{\mathrm{st}}(i,\tau)\) holds is at least
	\begin{equation*}
		1-	\exp\{-C(1-\epsilon)\lambda_0\ell^{1/3}\}
	\end{equation*}
	for some constant \(C\) and all large enough \(\ell\), which proves the claim.
\end{proof}

\begin{proof}[Proof of Claim~\ref{cl:1}]
	For each time $t\in [\tau\beta,\tau\beta+T]$, let $\Psi_t$ be the Poisson point process on $\mathbb{T}^d$ giving the locations at time $t$ of the particles that belong to $\Upsilon$, excluding all particles located at \(\rho(\tau\beta)\) at time \(\tau\beta\). Since the particles that start in \(Q^*\) move around and can leave \(Q^*\), we need to find a lower bound for the intensity of \(\Psi_t\) for times in \([\tau\beta,\tau\beta+T]\). Note that the distinguished particle \(x_0\) we are tracking is not part of \(\Psi\), since \(\Psi\) does not include particles located at \(\rho(\tau\beta)\) at time \(\tau\beta\). 
	
	We will need to apply heat kernel bounds from \cite[Theorem 2.2]{Hambly2009} to the particles in \(Q^*\), so we need to ensure that the time intervals we consider are large enough for the bounds to hold. 
	We will only consider times \(t\in[\ell^{4/3},T]\) so that for large enough \(\ell\), we have \(t\geq\sup_{\substack{x\in Q^*\\y\in Q^*}}\|x-y\|_1\) and so the heat kernel bounds from \cite[Theorem 2.2]{Hambly2009} hold. 
	Then, we have that for all sites \(x\in Q^*\) that are at least \(\ell\) away from the boundary of \(Q^*\) and at any such time \(t\) the intensity of \(\Psi_{\tau\beta+t}\) at vertex $x\in \mathbb{T}^d$ is at least
	\begin{equation*}
		\Psi_{\tau\beta+t}(x)\geq\sum_{\substack{y\in Q^*\\y\neq \rho(\tau\beta)}}(1-\epsilon)\lambda_0\mu_y\cdot \mathbb{P}_y[Y_t=x]
		= (1-\epsilon)\lambda_0\mu_x\sum_{\substack{y\in Q^*\\y\neq \rho(\tau\beta)}}\mathbb{P}_x[Y_t=y],
	\end{equation*}
	where $Y_t$ stands for the location of a simple random walk at time $t$, and $\mathbb{P}_y$ is the measure induced by a simple random walk starting from $y$. 
	In the last step above, we used that the simple random walk is reversible with respect to the measure $\mu$. 
	We now use the exit probability bound from \cite[Proposition 3.7]{Barlow2004} to get that 
	\[
		\sum_{\substack{y\in Q^*}}\mathbb{P}_x[Y_t=y]\geq 1-c_3\exp\{-c_4\ell^2/t\}.
	\]
	Next, we use \cite[Theorem 2.2]{Hambly2009} to account for the particles at \(\rho(\tau\beta)\), yielding 
	\begin{equation*}
		\sum_{\substack{y\in Q^*\\y\neq \rho(\tau\beta)}}\mathbb{P}_x[Y_t=y]\geq 1-c_3\exp\left\{-c_4\ell^2/t\right\}-C_Mc_5t^{-d/2}.
	\end{equation*}
	This gives that for any \(t\in[\ell^{4/3},T]\), the intensity of \(\Psi_{\tau\beta+t}\) is at least
	\[
		\Psi_{\tau\beta+t}(x)\geq (1-\epsilon)\lambda_0 \mu_x(1-c_3\exp\{-c_4\ell^2/T\}-C_Mc_5\ell^{-2d/3}).
	\]

	Let \([\tau\beta,\tau\beta+T]\) be divided into subintervals of length \(W\in(0,T]\), where we set \(W=\ell^{4/3}\) so that it is large enough to allow the use of the heat kernel bounds from \cite[Theorem 2.2]{Hambly2009}. Let \(J=\{1,\dots, \lfloor T/W\rfloor \}\) and \(t_j:=\tau\beta+jW\). 
	Then the intensity of particles that share a site with the distinguished particle \(x_0\) at least once among times \(\{t_1, t_2, \dots, t_{ \lfloor T/W\rfloor }\}\) is at least
	\begin{align*}
		&\sum_{j\in J}\Psi_{t_j}(\rho(t_j))\mathbb{P}_{\rho(t_j)}[Y_{r-t_j}\neq \rho(r) \;\forall r\in\{t_{j+1},\dots,t_{\lfloor T/W\rfloor}\}]\\
		&\geq (1-\epsilon)\lambda_0 C_M^{-1}(1-c_3\exp\{-c_4\ell^2/T\}-C_Mc_5\ell^{-2d/3})\sum_{j\in J}\left(1-\sum_{z>j}\mathbb{P}_{\rho(t_j)}[Y_{t_z-t_j}= \rho(t_z)]\right).
	\end{align*}

	We want to make all of the terms of the sum over \(J\) positive, so we consider the term \(\sum_{z>j}\mathbb{P}_{\rho(t_j)}[X_{t_z-t_j}= \rho(t_z)]\) and show that it is smaller than \(\frac{1}{2}\) for large enough \(\ell\). 
	To do this, we use \cite[Theorem 2.2]{Hambly2009}, which hold when \(W\geq\ell^{4/3}\) and \(\ell\) is large enough, to bound it from above by
	\begin{align}
		\sum_{z>j}\mathbb{P}_{\rho(t_j)}[Y_{t_z-t_j}= \rho(t_z)]
		&\leq \sum_{z>j}C_MC_{HK}(t_z-t_j)^{-d/2}\nonumber\\
		&\leq C_M C_{HK}W^{-d/2}\sum_{z=1}^{T/W-j}z^{-d/2}\label{for:TW}
	\end{align}
	where \(C_{HK}\) is a constant coming from \cite[Theorem 2.2]{Hambly2009}. Then, (\ref{for:TW}) can be bounded from above by
	\begin{equation}
		C_M C_{HK} W^{-d/2}\left(2+\sum_{z=3}^{T/W-j}z^{-d/2}\right)
		\leq C_M C_{HK}W^{-d/2}\left(2+\int_{2}^{T/W}z^{-d/2}dz\right)\label{for:TWintegral}.
	\end{equation}
	Let \(C\) be a constant that can depend on \(C_{HK}\), \(C_M\) and \(d\). Then for \(d=2\), (\ref{for:TWintegral}) is smaller than \(CW^{-1}\log(T/W)\), and for \(d\geq 3\) the expression in (\ref{for:TWintegral}) is smaller than \(CW^{-d/2}\). Thus, setting \(\ell\) large enough, both terms are smaller than \(\frac{1}{2}\).

	Then, as a sum of Poisson random variables, we get that \(\Upsilon'\) is a Poisson random variable with a mean at least
	\[
		(1-\epsilon)\lambda_0 C_M^{-1}(1-c_3\exp\{-2c_4\ell^2/T\}-C_Mc_5\ell^{-2d/3})\tfrac{T}{2W}.
	\]
	Using that \(T=\ell^{5/3}\) and setting \(\ell\) large enough establishes the claim, with \(C_1\) being any constant satisfying \(C_1<\frac{C_M^{-1}}{2}\).
\end{proof}	

\begin{proof}[Proof of Claim~\ref{cl:2}]
	We now prove that for large enough \(\ell\), if there are \(N\) particles inside of \(Q^*\) at time \(\tau\beta+T\), there is at least one of them inside \(Q^{**}\) at time \((\tau+1)\beta\) with probability at least
	\(
	1-\exp\{-Nc_p\}.
	\)

	For \(t^{2/3}\geq\sup_{\substack{x\in Q^*\\y\in Q^{**}}}\|x-y\|_1\), define \(p_t:=\inf_{\substack{x\in Q^*}}\sum_{y\in Q^{**}}\mathbb{P}_x[Y_t=y]\). Then, if we define \(\mathrm{bin}(N,p_t)\) to be a binomial random variable with parameters \(N\in\mathbb{N}\) and \(p_t\in[0,1]\), it directly follows that we can bound probability of one of the \(N\) particles from \(Q^*\) being inside \(Q^{**}\) at time \((\tau+1)\beta\) from below by
	\[
		\mathbb{P}[\mathrm{bin}(N,p_t)\geq1]\geq1-\exp\{-Np_t\}.
	\]

	It remains to show that for \(t=\beta-T\), we have that \(p_t\geq c_p>0\) for some constant \(c_p\). We will again use the heat kernel bounds from \cite[Theorem 2.2]{Hambly2009} for the pair \(x,y\), which hold if \(\|x-y\|_1^{3/2}\leq\beta-T\) for all \(x\in Q^*,y\in Q^{**}\). Given the ratio \(\beta/\ell^2\), \(d\) and \(\eta\), this is satisfied if \(\ell\) is large enough. Then we have that
	\begin{align*}
		p_{\beta-T}=\inf_{x\in Q^*}\sum_{y\in Q^{**}}\mathbb{P}_x[Y_{\beta-T}=y]
		\geq\inf_{x\in Q^*}C_{M}^{-1}\sum_{y\in Q^{**}}c_1 \beta^{-d/2}\exp\left\{-c_2\frac{\|x-y\|_1^2}{{\beta-T}}\right\}.
	\end{align*}
	Now we use that \(x\) and \(y\) can be at most \(c_\eta\ell\) apart where \(c_\eta\) is a constant depending on \(d\) and \(\eta\) only, and that \(\beta-T\geq\beta/2\) for \(\ell\) large enough. Hence,
	\begin{align*}
		p_{{\beta-T}}&\geq \inf_{x\in Q^*}C_M^{-1}\sum_{y\in Q^{**}}c_1 \beta^{-d/2}\exp\left\{-c_2\frac{2(c_\eta\ell)^2}{\beta}\right\}\\
		&= C_M^{-1}c_1\ell^d\left(\frac{1}{\beta}\right)^{d/2}\exp\left\{-c_2\frac{2(c_\eta\ell)^2}{\beta}\right\}\\
		&\geq c_p.
	\end{align*}
\end{proof}
 
Lemma \ref{prop:event} implies that for \(\ell\) large enough, by setting \(\eta\) to be a large enough constant, and defining the increasing event \(E_{\textrm{st}}\) as in Definition~\ref{def:mainEvent}, 
the information spreads among neighboring cells. 
Since the Lipschitz net surrounds the origin at distance $O(\log^2n)$, we have that in at most poly-logarithmic time, 
the initially informed particle will enter some cell \(\prod_{j=1}^d[i_j\ell,(i_j+1)\ell]\) for which \((i,\tau)\) is in some Lipschitz surface \(F\) of $F_\mathrm{net}$. 
Once that holds, we know that the event \(E_{\mathrm{st}}(i,\tau)\) occurs. 
By the definition of \(E_{\mathrm{st}}(i,\tau)\), we obtain that the initially informed particle in \((i,\tau)\) informs other particles causing 
the information to spread to each \((i',\tau+1)\) for which \(\|i'-i\|_{\infty}\leq \eta\).
	
Let \((b,h)\) be the base-height index of the cell \((i,\tau)\in F\). 
Recall that \(h\) is one of the spatial dimensions. 
We will also select one of the \(d-1\) spatial dimensions from \(b\) and denote it \(b_1\). 
Let \(b'\in\mathbb{T}^d_*\) be obtained from \(b\) by increasing the time dimension from 
\(\tau\) to \(\tau+1\), and by increasing the chosen spatial dimension from \(b_1\) to \(b_1+1\). 
Since \(\|b-b'\|_1=2\), we can choose \(h'\in\mathbb{T}\) such that \((b',h')\in F\) and \(|h-h'|\leq 2\), where the latter holds by the Lipschitz property of \(F\). 
Therefore, there must exists \(i'\in\mathbb{T}^d\) such that \((i',\tau+1)\) is the space-time cell corresponding to \((b',h')\) and \(\|i-i'\|_{\infty}\leq 4\). 
Hence, at time \((\tau+1)\beta\), there is an informed particle in the cube indexed by \(i'\) if $\eta$ is at least $4$ and $E_\mathrm{st}(i,\tau)$ holds.

Using this mechanism, we can show that after some time of order \(n\), the information has spread along the surfaces across the entire torus. 
\begin{lemma}\label{prop:phase1}
	Let \(F_{\textrm{net}}\) be the Lipschitz net with constant \(C_0\) which surrounds the origin at distance $O(\log^2n)$. 
	There exists a constant \(C_T>0\), independent of \(C_0\), such that for every \((i,\tau)\in F_{\textrm{net}}\) for which \(\tau\beta\geq C_Tn\), 
	there is at least one informed particle inside the cube \(\prod_{j=1}^d[(i_j-\eta+1)\ell,(i_j+\eta)\ell]\) for all times in \([\tau\beta,(\tau+1)\beta]\).
\end{lemma}

\begin{proof}
	Let \(E_{\mathrm{st}}(i,\tau)\) be defined as in Lemma \ref{prop:event} and let \(F_{\textrm{net}}\) be the Lipschitz net with constant \(C_0\), 
	corresponding to the event \(E_{\mathrm{st}}(i,\tau)\).
	We have by the fact that \( F_\mathrm{net}\) surrounds the origin at a distance \(O(\log^2n)\) and that each cell represents a time interval of length \(\beta\), 
	that it takes at most $O(\beta \log^2n)$ time for the information to enter \( F_\mathrm{net}\).
	Once the informed particle is in a space-time cell of some surface \( F_{k}^q\) of $F_\mathrm{net}$, 
	we have by the definition of \(E_{\mathrm{st}}(i,\tau)\) with \(\eta=d\), 
	that it takes at most \(\frac{2n}{\ell}\) steps for the information to spread across the surface (moving between neighboring cells), 
	so that all space-time cells \((i,\tau)\in  F_k^q\) for which \(\tau=\frac{2n}{\ell}+O(\log^2 n)\) contain an informed particle.
	
	Next, for any $q',k'$ with $q'\neq q$, we know by Lemma \ref{lem:change_surface} that for any $\tau$ there are neighboring cells $(i,\tau)\in F_k^q$ and $(i',\tau)\in F_{k'}^{q'}$. 
	Therefore, it takes at most \(\beta\) time for the information to enter any surface $F_{k'}^{q'}$ with $q' \neq q$,
	and another $\frac{2n}{\ell}\beta$ amount of time to spread to all cells in those surfaces, so that all cells \((i,\tau)\in  F_{k'}^{q'}\) 
	for which \(\tau=\frac{4n}{\ell}+1+O(\log^2n)\) contains an informed particle.
	It still remains to spread the information to the surfaces $F_{k'}^q$ with $k'\neq k$. Again, this takes at most $\frac{2n}{\ell}\beta+\beta$ time by the same argument above.
   Putting everything toghether, we obtain that for any \(k\), any \(q\), and all \((i,\tau)\in  F_k^q\) for which 
   \(\tau\beta\geq C_Tn\geq (\frac{6n}{\ell}+2+O(\log^2n))\beta\), 
   where we set \(C_T\) large enough for the second inequality to hold, there is at least one informed particle in the cube \(\prod_{j=1}^d[(i_j-\eta+1)\ell,(i_j+\eta)\ell]\) for all times in \([\tau\beta,(\tau+1)\beta]\).
\end{proof}

Using Lemma \ref{prop:phase1} and the geometric properties of the Lipschitz net, we can show that there is a density of informed particles everywhere on the torus for an interval of time of order \(n\).

\begin{thrm}\label{prop:phase1density}
	There exists constants \(C_{\beta}\geq 1\) and \(C_{\ell}>0\) such that the following holds.
	Let \(C_T\) be the constant from Lemma \ref{prop:phase1}. Tessellate \(\mathbb{T}^d\) into cubes \((Q_m)_m\) of side length \(C_{\ell}\log^3(n)\). Then, for all times  \(t\in[C_Tn,(C_T+C_{\beta})n]\), there is at least one informed particle in each subcube \(Q_m\) with probability at least
	$
		1-n^{-\omega(1)}.
	$
\end{thrm}

\begin{proof}
	Fix \(\ell\) sufficiently large for Lemma \ref{prop:event} and Theorem \ref{thrm:surface_event_simple} to hold and recall that the ratio \(\beta/\ell^2\) is fixed.
	Let also \(n\gg \ell\). 
	Then, there exists a constant \(C_T\) so that, for any large enough choice of \(C_0\), Lemma \ref{prop:phase1} gives that 
	for every space-time cell \((i,\tau)\) of the Lipschitz net \(F_{\textrm{net}}\) that satisfies \(\tau\beta\geq C_Tn\), 
	there is at least one informed particle in the region \(\prod_{j=1}^d[(i_j-\eta+1)\ell,(i_j+\eta)\ell]\) at all times in \([\tau\beta,(\tau+1)\beta]\). 
	We can, without loss of generality, assume \(C_T\) is such that \(C_Tn=\beta\tau^*\) for some \(\tau^*\in\mathbb{N}\). Then, we only have to show that for all cubes \(Q_m\) of side length \(C_{\ell}\log^3(n)\), there exist space-time cells \((i,\tau)\) such that the region \(\prod_{j=1}^d[(i_j-\eta+1)\ell,(i_j+\eta)\ell]\) is contained in \(Q_m\) and such that \([C_Tn,(C_T+C_{\beta}) n]\subseteq\bigcup_{\tau}[\tau\beta,(\tau+1)\beta]\), where \(C_{\beta}\) is a constant greater or equal to \(1\). 
	
	Let \(C_{\beta}=k^*\beta\) where \(k^*\) is the smallest integer for which \(k^*\beta\geq 1\) and fix the Lipschitz net constant \(C_0\) to be greater or equal to \((C_T+C_{\beta})\ell/\beta\). Then, we have from Theorem \ref{thrm:net} that the Lipschitz net with constant \(C_0\) exists with probability at least 
	$
		1-n^{-\omega(1)}.
	$
	
	We now show that if this Lipschitz net exists, the lemma holds.
	Let \( F_s^q\) and \( F_{s+1}^q\) be any two consecutive two-sided surfaces of the Lipschitz net and let \((b,h)\in  F_s^q\) and \((b,h')\in F_{s+1}^q\) be two base-height cells with the same base.
	By definition of the Lipschitz net, we have that the height of each Lipschitz surface in the net is at most \(\frac{\log^3(n/\ell)}{2}\) for all space-time cells that satisfy \(\tau\in\{0,1,\dots,C_0n/\ell\}\). 
	Since the base-height cells \((b,h)\) and \((b,h')\) might belong to opposite sides of the two-sided Lipschitz surfaces, we therefore have that \(|h-h'|\leq 2\log^3(n/\ell)\) for all base-height cells for which \(\tau\beta<(C_T+C_{\beta})n\leq C_0\beta n/\ell\). Note that this holds for all \(q\in\{1,\dots,d\}\) and 
	recall  that by Lemma \ref{prop:phase1} there is an informed particle inside the region \(\prod_{j=1}^d[(i_j-\eta+1)\ell,(i_j+\eta-1)\ell]\) throughout the entire time interval \([\tau\beta,(\tau+1)\beta]\). 
	Therefore, for every cube of side length at least \(2\ell\log^3(n/\ell)+2\eta\ell\) on the torus and throughout every time interval of the form above, there is at least one informed particle inside the cube.
	By repeating this argument for all \(\tau\) that satisfy \(\tau\beta\in[C_Tn,(C_T+C_{\beta})n)\), we have that this holds for the entire time interval \([C_Tn,(C_T+C_{\beta})n]\).
\end{proof}


Before turning to the proof of Theorem~\ref{thrm:total}, we state a theorem that gives that if we start with a density of particles on a cube, 
regardless of how they are placed inside some subcubes, 
we can couple their positions after some time with a Poisson point process that is independent of their initial locations. 
This gives a type of local mixing property for random walks on $\mathbb{T}^d$ with i.i.d.\ conductances. 
For the proof of this technical result, refer to \cite[Theorem 3.1]{Gracar2016}.

\begin{thrm}\label{thrm:mixing}
	Let \(G\) be a uniformly elliptic graph with edge weights \(\mu_{x,y}\). There exist constants \(c_0\), \(c_1\), \(C>0\) such that the following holds.
	Fix \(K>\ell>0\) and \(\epsilon\in(0,1)\). Consider the cube \(Q_K\) tessellated into subcubes \((T_i)_{i}\) of side length \(\ell\) and assume that $\ell$ is large enough.  
	Let \((x_j)_{j}\subset  Q_{K}\) be the locations at time \(0\) of a collection of particles, such that each subcube \( T_i\) contains at least \(\sum_{y\in  T_i}\beta\mu_y\) particles for some \(\beta>0\).
	Let \(\Delta\geq c_0\ell^2\epsilon^{-4/\Theta}\) where \(\Theta\) is a constant that depends on the weight bounds.
	For each \(j\) denote by \(Y_j\) the location of the \(j\)-th particle at time \(\Delta\). 
	Fix \(K'>0\) such that \(K-K'\geq\sqrt{\Delta}c_1\epsilon^{-1/d}\). Then there exists a coupling \(\mathbb{Q}\) of an independent Poisson point process \(\psi\) with intensity measure \(\zeta(y)=\beta(1-\epsilon)\mu_y\), \(y\in \mathcal{C}_{\infty}\), and \((Y_j)_{j}\) such that within \( Q_{K'}\subset  Q_K\), \(\psi\) is a subset of \((Y_j)_{j}\) with probability at least
	\[
		1-\sum_{y\in  Q_{K'}}\exp\left\{-C\beta\mu_y\epsilon^2\Delta^{d/2}\right\}.
	\]
\end{thrm}

\subsection{Proof of Theorem \ref{thrm:total}}
\begin{proof}
	Let \(C_T\) be the constant from Lemma \ref{prop:phase1} and let \(C_{\beta}\) and \(C_{\ell}\) be the constants from Theorem \ref{prop:phase1density}. 
	We want to bound the probability that at time \((C_T+C_{\beta})n\) there is at least one particle on \(\mathbb{T}^d\) that is not informed. 
	By using that the particles on \(\mathbb{T}^d\) form a Poisson point process with intensity \(\lambda(y)=\lambda_0\mu_y\), we have that this probability can be bounded from above by
	$
		\lambda_0p_n\sum_{y\in\mathbb{T}^d}\mu_y,
	$
	where \(p_n\) is an upper bound for the probability that a single particle is not informed by time \((C_T+C_{\beta})n\) on the torus of side length \(n\), uniformly on the initial location of the particle. 
	We now proceed to find the bound \(p_n\).
	
	Let \(F\) be the event that a particle located somewhere on the torus does not become informed during \([C_Tn,(C_T+C_{\beta})n]\). Note that the probability that a particle does not get informed by time \((C_T+C_{\beta})n\) is smaller than the probability of \(F\), so \(p_n\leq\mathbb{P}[F]\). Let \(t\in(0,C_{\beta}n)\) be a time step we will fix later and consider the time interval \([C_Tn,(C_T+C_{\beta})n]\) split into subintervals of length \(t\), i.e.\ let the interval be split into subintervals of the form \([C_Tn+kt,C_Tn+(k+1)t]\) for \(k\in\{0,1,\dots,\lfloor C_{\beta}n/t\rfloor-1\}\). Let \(F_k\) denote the event that a particle located somewhere on the torus does not become informed during the time interval \([C_Tn+kt,C_Tn+(k+1)t]\). We then have that
	\[
		\mathbb{P}[F]\leq\mathbb{P}[F_0\cap F_1\cap\dots\cap F_{\lfloor C_{\beta}n/t\rfloor-1}].
	\]
	
	Tessellate \(\mathbb{T}^d\) into cubes \((Q_i)_i\) of side length \(C_{\ell}\log^3(n)\), indexed by \(i\). Let \(D_k\) be the event that time \(C_Tn+kt\) there is at least one informed particle in every cube \(Q_i\). We can then write
	\begin{equation}
		\mathbb{P}[F_0\cap F_1\cap\dots\cap F_{\lfloor C_{\beta}n/t\rfloor-1}]
	\leq \mathbb{P}\left[\bigcap_{k=0}^{\lfloor C_{\beta}n/t\rfloor-1}\left(F_{k}\cap D_{k}\right)\right]
	+\mathbb{P}\left[\bigcup_{k=0}^{\lfloor C_{\beta}n/t\rfloor-1}D^\mathsf{c}_{k}\right]\label{eq:density_noDensity}.
	\end{equation}
	
	To bound the second term, we apply Theorem \ref{prop:phase1density}, which gives that there is at least one informed particle in every cube \(Q_i\) of side length \(C_{\ell}\log^3(n)\) for all times during \(t\in[C_Tn,(C_T+C_{\beta})n]\) with high probability. Therefore, it holds that
	\begin{equation}\label{eq:densityBound}
		\mathbb{P}\left[\bigcup_{k=0}^{\lfloor C_{\beta}n/t\rfloor-1}D^\mathsf{c}_{k}\right]= n^{-\omega(1)}.	
	\end{equation}
	
	We now focus on the first term of (\ref{eq:density_noDensity}). By rearranging the expression inside the probability and using the chain rule, we have that
	\[
		\mathbb{P}\left[\left(\bigcap_{k=0}^{\lfloor C_{\beta}n/t\rfloor-1}F_{k}\right)\cap\left(\bigcap_{k=0}^{\lfloor C_{\beta}n/t\rfloor-1}D_{k}\right)\right]\leq \mathbb{P}[F_{0}\cap D_0]\prod_{k=1}^{\lfloor C_{\beta}n/t\rfloor-1}\mathbb{P}\left[F_{k}\cap D_{k}\;\Big|\;\bigcap_{j<k}F_{j}\cap D_{j}\right].
	\]
	In order to bound the terms \(\mathbb{P}\left[F_{k}\cap D_{k}\;\middle|\;\bigcap_{j<k}F_{j}\cap D_{j}\right]\), first note that 
	\begin{align*}
		\mathbb{P}\left[F_{k}\cap D_{k}\;\Big|\;\bigcap_{j<k}F_{j}\cap D_{j}\right]&\leq\mathbb{P}\left[F_{k}\;\Big|\;D_{k}\cap\bigcap_{j<k}F_{j}\cap D_{j}\right]\mathbb{P}\left[D_{k}\;\Big|\;\bigcap_{j<k}F_{j}\cap D_{j}\right]\\
		&\leq \mathbb{P}\left[F_{k}\;\Big|\;D_{k}\cap\bigcap_{j<k}F_{j}\cap D_{j}\right],
	\end{align*}
	and similarly, 
	$
		\mathbb{P}[F_{0}\cap D_0]=\mathbb{P}[F_{0}\;|\; D_0]\mathbb{P}[D_0]\leq \mathbb{P}[F_{0}\;|\; D_0].
	$
	
	Next, we show a bound for \(\mathbb{P}\left[F_{k}\;\Big|\;D_{k}\cap\bigcap_{j<k}F_{j}\cap D_{j}\right]\) that holds uniformly on all configurations for which \(D_k\) holds.
	We do this by applying Theorem \ref{thrm:mixing} to find a uniform bound on the probability of a particle remaining uninformed, given there is a density of informed particles on the torus \(\mathbb{T}^d\) at the beginning of the time interval we consider. More precisely, we set the terms of Theorem \ref{thrm:mixing} as follows, where we mark them with a bar to help distinguish them from other terms in this proof. Let
		\(\bar K=n\), 
		\(\bar \ell=C_{\ell}\log^3(n)\), and 
		\(\bar \epsilon=\frac{1}{2}\). Let  
		\(\bar \Delta=C_{\Theta}\log^8(n)\), where \(C_{\Theta}\) is a constant sufficiently large for \(\bar \Delta\) to satisfy the conditions of Theorem \ref{thrm:mixing} for all \(n\). We fix the time step \(t\) to be equal to \(\bar \Delta\) and let
		\(\bar K'=n-C_{\bar\epsilon}\sqrt{\bar\Delta}\), where \(C_{\bar\epsilon}=c_1\bar\epsilon^{-1/d}\).
	We now have by the definition of \(D_k\) for every \(k\in\{0,1,\dots,\lfloor C_{\beta}n/\bar\Delta\rfloor-1\}\) that at time \(C_Tn+kt\) there is at least one informed particle in every subcube \(Q_i\), so there are at least 
	\[
	\frac{1}{C_MdC_{\ell}^d\log^{3d}(n)}\sum_{y\in Q_i}\mu_y
	\]
	informed particles in every cube. We set the parameter \(\bar\beta\) from Theorem \ref{thrm:mixing} to be \(\bar\beta=\frac{1}{C_MdC_{\ell}^d\log^{3d}(n)}\) and apply the theorem. This gives us that after the informed particles move around for time \(\bar\Delta\), they stochastically dominate a Poisson point process of intensity \(\bar\zeta(y)=\frac{1}{2}\frac{1}{C_MdC_{\ell}^{d}\log^{3d}(n)}\mu_y\) inside the cube of side length \(\bar K'\). Using (\ref{eq:mu_bounds_new}), we have that this coupling fails with probability at most
	\begin{equation}\label{eq:couplingFails}
	\sum_{y\in Q_{K'}}\exp\left\{-C\frac{1}{4}\frac{1}{C_MdC_{\ell}^{d}\log^{3d}(n)}C_{\Theta}^{d/2}\log^{4d}(n)\mu_y\right\}\leq n^d\exp\{-C_1\log^d(n)\},
	\end{equation}
	where \(C\) is the constant from Theorem \ref{thrm:mixing} and \(C_1\) is some constant that depends on \(d\).
	Note that this bound only depends on the size of \(Q_{K'}\) and as such is independent of the site the cube is centered around.
	
	Next, if \(D_k\) holds and the coupling succeeds, the number of informed particles at a given site \(y\) of the torus at time \(C_Tn+(k+1)t\) stochastically dominates a Poisson random variable of intensity \(\frac{1}{2C_MdC_{\ell}^{d}\log^{3d}(n)}\mu_y\). Since the probability that a particle is not informed during the interval \([C_Tn+kt,C_Tn+(k+1)t]\) is smaller than the probability of not getting the information only at the end of the interval, we have that \(\mathbb{P}[F_k\;|\;\{\textrm{coupling succeeds}\}\cap D_k]\) can be bound by the probability that at the end of the time interval, there are no informed particles at the location of the particle we are considering. Using (\ref{eq:mu_bounds_new}) to bound \(\mu_y\), we have for some constant \(C_2\) that \(\mathbb{P}[F_k\;|\;\{\textrm{coupling succeeds}\}\cap D_k]\) is at most the probability that a Poisson random variable with intensity \(\frac{C_2}{\log^{3d}(n)}\) is \(0\), i.e.
	\begin{equation}\label{eq:noParticle}
		\mathbb{P}[F_k\;|\;\{\textrm{coupling succeeds}\}\cap D_k]\leq\exp\left\{-\frac{C_2}{\log^{3d}(n)}\right\}.
	\end{equation}
	This bound holds uniformly across all sites of the torus where the particle might be located and across all configurations of particles for which \(D_k\) holds.
	Combining (\ref{eq:couplingFails}) and (\ref{eq:noParticle}) we therefore have for all \(k\in\{0,1,\dots,\lfloor C_{\beta}n/t\rfloor-1\}\) that 
	\begin{equation*}
		\mathbb{P}\left[F_{k}\;\Big|\;D_{k}\cap\bigcap_{j<k}F_{j}\cap D_{j}\right]
		\leq n^d\exp\{-C_1\log^d(n)\}+\exp\left\{-\tfrac{C_2}{\log^{3d}(n)}\right\}.
	\end{equation*}
	
	Using the definition of \(t\), the bound from (\ref{eq:densityBound}) and applying the above bound for all \(k\in\{0,1,\dots,\lfloor C_{\beta}n/t\rfloor-1\}\), we have that \(\mathbb{P}[F_0\cap F_1\cap\dots\cap F_{\lfloor C_{\beta}n/t\rfloor-1}]\) from (\ref{eq:density_noDensity}) is smaller than
	\begin{equation*}
		\left(n^d\exp\{-C_1\log^d(n)\}\right)^{C_{\beta}n/(C_{\Theta}\log^8(n))}+\exp\left\{-\tfrac{C_2C_{\beta}n}{C_{\Theta}\log^{3d+8}(n)}\right\}+n^{-\omega(1)}.
	\end{equation*}
	
	Using that \(p_n\leq \mathbb{P}[F_0\cap F_1\cap\dots\cap F_{\lfloor C_{\beta}n/t\rfloor-1}]\) 
	and \(\mu_y
	\leq C_Md\) by (\ref{eq:mu_bounds_new}), we get that the probability that there exists a particle that has not been informed by time \((C_T+C_{\beta})n\) is at most
	\begin{equation*}
		C_Md\lambda_0n^d\bigg(\left(n^d\exp\{-C_1\log^d(n)\}\right)^{C_{\beta}n/(C_{\Theta}\log^8(n))}+\exp\left\{-\tfrac{C_2C_{\beta}n}{C_{\Theta}\log^{3d+8}(n)}\right\}+n^{-\omega(1)}\bigg).
	\end{equation*}
	Since the above is $n^{-\omega(1)}$, the proof is completed.
\end{proof}
%
%

\section{Conclusion}
We have established a tight bound on the flooding time (up to constant factors) for the spread of information between random walk particles on the discrete torus of size \(n\), equipped with i.i.d., uniformly elliptic conductances. 
To prove this, we develop a framework to control dependences, which given any increasing, local event that is likely enough, 
one can find a Lipschitz surface and a Lipschitz net through space-time where this event holds. 
We believe this result can be applicable to analyze other processes and algorithms on systems of random walk particles. 
We also believe that this framework can be adapted to work with different types of particle systems, for example, when the particles do not move independently of one another, but nonetheless obey some 
local mixing. 

\appendix
\section{Standard large deviation results}
\begin{lemma}[Chernoff bound for Poisson]\label{lem:chernoff}
Let \(P\) be a Poisson random variable with mean \(\lambda\). Then, for any \(0<\epsilon<1\),
\[
	\mathbb{P}[P<(1-\epsilon)\lambda] < \exp\{-\lambda\epsilon^2/2\}
\quad \text{and} \quad
	\mathbb{P}[P > (1 + \epsilon)\lambda] < \exp\{-\lambda\epsilon^2/4\}.
\]
\end{lemma}

\addcontentsline{toc}{section}{References}
\bibliography{library}

\end{document}